\documentclass[12pt, titlepage]{article}
\usepackage{eurosym}
\usepackage{amsmath,amssymb,amsthm,amsfonts}
\usepackage{epstopdf}
\usepackage{subcaption}
\usepackage{appendix}
\usepackage{amsmath}
\usepackage{amsfonts}
\usepackage{amssymb}
\usepackage{multirow}
\usepackage{rotating}
\usepackage{xcolor}
\usepackage{lscape}
\usepackage{verbatim}
\usepackage{cancel}

\numberwithin{equation}{section} 

\newtheorem{theorem}{Theorem}[section]
\newtheorem{corollary}{Corollary}[section]

\newtheorem{proposition}{Proposition}[section]
\newtheorem{definition}{Definition}[section]
\newtheorem{remark}{Remark}
\newtheorem{example}{Example}

\def\xx{\boldsymbol{x}}

\def\SS{\boldsymbol{S}}

\def\BB{\boldsymbol{B}}

\def\NN{\mathbb{N}}

\def\RR{\mathbb{R}}

\def\l{\mathcal{L}}

\def\dym{\mathcal{D}(\mathcal{Y},\mu)}
\def\dy{\mathcal{D}(\mathcal{Y},1)}

\DeclareMathOperator{\rank}{rank}

\begin{document}

\author{ R.FONTANA \\ \textit{\ 
Department of Mathematical Sciences G. Lagrange,} \\ {Politecnico di
Torino.}
\\ P. SEMERARO \\ \textit{\ 
Department of Mathematical Sciences G. Lagrange,} \\ {Politecnico di
Torino.}}
\title{Measuring distribution risk in discrete models}
\maketitle

\begin{abstract}
Model risk measures consequences of choosing a model in a class of possible alternatives.
We find analytical and simulated bounds for payoff functions on classes of plausible alternatives of a given discrete  model.
We  measure  the impact of choosing a risk-neutral measure on convex derivative pricing in incomplete markets. We find analytical bounds for prices of European and American options in the class of all risk-neutral measures, and we also find simulated bounds for given classes of perturbations of the minimal martingale equivalent measure.

\noindent \textbf{keywords}: Model risk, discrete distribution, convex polytope, discrete pricing models, incomplete multinomial models, risk-neutral evaluation.
\end{abstract}

\section{Introduction}

Mathematical models are essential in order to develop quantitative methods in finance that are used, for example, in  pricing, risk management and portfolio selection. As a consequence, model risk cannot be ignored and must be quantified.
A statistical model in finance has two main ingredients: the distribution of a random source, and a payoff function describing how risk factors impact on the quantities of interest (\cite{breuer2016measuring}).  In this paper we consider the payoff function as given and focus on what \cite{breuer2016measuring} refers to as distribution model risk, that is the risk associated with a wrong choice of a random source distribution.

Distribution model risk aims to quantify the impact on a payoff function  of working with a wrong  distribution in the class of plausible distributions for a random phenomenon.
Given a random variable $Y$ in a space of random variables $\chi$,  let  $\rho:\chi\rightarrow \RR$ be  a law invariant functional, i.e. if $X$  and $Y$ have the same distribution then $\rho(X)=\rho(Y)$ (\cite{bellini2021law}). The functional $\rho$ can be  a measure of risk as well as the payoff of a financial derivative. The cost of a wrong choice of the distribution of $Y$ in a class $\mathcal{P}$  of plausible distributions for $Y$ can be measured as the range spanned by $\rho(Y)$ across $\mathcal{P}$.
As a consequence, the identification of a class of plausible alternatives for the distribution of $Y$ and the ability to find a range for a functional defined on it are the two key issues to address in distribution model risk.

This paper studies distribution model risk for all the cases where the source of risk is a random variable $Y$ with discrete support  included in  $\mathcal{Y}=\{y_1,\ldots, y_d\}$ .
 As the largest plausible class of  discrete distributions for $Y$ we consider the class $\dym$  of distributions with support included in  $\mathcal{Y}=\{y_1,\ldots, y_d\}$ and with given mean $\mu$. We therefore do not consider the risk associated with the mean.

The mathematical foundation of this work is the   geometrical representation of the class $\dym$ as a convex polytope and the analytical closed form of its extremal points. We  provide this  representation by generalizing  the results in \cite{fontana2021model}.  Furthermore, we find  analytical bounds in $\dym$ of a class of  functionals $\rho$ commonly used in pricing and risk measurement. The range spanned by a functional $\rho$ across $\dym$ is analytical and is a measure of the distribution model risk  associated with the whole class $\dym$.
However, in some cases the class $\dym$,  the range spanned by $\rho$ or both of these are too wide. The geometrical characterization of $\dym$ allows us to overcome these issues by simulation. Indeed, the full characterization of the class allows us not only to simulate from the class, but also to simulate from subclasses properly defined. Given a subclass $\mathcal{P}_Q\subset \dym$ of plausible alternatives for $Q\in \dym$ we can find empirically the distribution across $\mathcal{P}_Q$  and the bounds in $\mathcal{P}_Q$ of a wide class of functionals.
The problem of defining a proper subclass of alternative distributions to a given $Q$ does not have a theoretical solution and depends on the specific context. A possible choice that we consider is the one proposed in \cite{breuer2016measuring} to take a ball of distributions, defined in terms of relative entropy, centered at $Q$.

 In this paper we have two applications in mind: the first is pricing in incomplete models and  the second is  credit portfolio management.
The latter  application is an extension of the model risk analysis performed in \cite{fontana2021model}, where we relax the assumption of an equally weighted portfolio.  Removing the assumption of equal weights, the loss $L$ has a probability mass function in $\dym$.  The results discussed here can be used to measure the risk associated with a wrong choice of the  joint default distribution.
However, we  leave the development of  this application  to future research and  we focus on the application to pricing.

When closed form solutions for derivative prices are not available numerical methods have to be used. One method consists of approximating the underlying dynamics. In the classical Cox, Ross and Rubinstein Binomial model \cite{cox1979option}, the lognormal price dynamics is approximated  with a binomial tree. Alternative approximations based on multinomial lattices are proposed, for example in \cite{kamrad1991multinomial}, to account for asymmetry and kurtosis of logreturns, as discussed in \cite{ssebugenyi2013minimal}.  Multinomial lattices  are workable models used to price derivatives, including exotic derivatives, but they lead to incomplete markets. As discussed in  \cite{cont2006model} there is a risk associated with the choice of a risk-neutral martingale measure used to price, since a wrong choice results in  wrong prices. Our knowledge of the class of discrete distributions allows us to measure the impact on prices of a wrong choice of the martingale equivalent pricing measure in a discrete pricing model.
We consider a market $\mathcal{M}$ with a risky asset and a riskless one and we  model the risky asset dynamics with a
multinational recombining tree to account for asymmetry and kurtosis in  logreturns, as discussed in \cite{ssebugenyi2013minimal}.

Our main result is to find analytical  bounds for the prices of convex derivatives in incomplete markets, using the geometrical structure of the class of risk-neutral probabilities. Furthermore, by uniform simulation we can find empirical bounds in smaller classes of plausible choice of risk-neutral measures, for example a class of perturbations of the minimal martingale equivalent measure.
We provide explicit analytical bounds for the case of European call and put options and for the more interesting case of American put options.
We also provide a first application example on real data.

The paper is organized as follows. Section \ref{Sec:distrrisk}  provides the geometrical representation of the class of discrete distributions with given mean and provides analytical bounds for a wide class of functional $\rho$.  These results are applied to pricing models in Section \ref{Distr:price} that provide analytical and empirical  bounds for derivative prices in multinomial models. Although the construction of an approximating tree for a continuous model is beyond the scope of this work, Section \ref{real} presents a simple application on real data, where the multinomial parameters are chosen to fit the sample  moments of real logreturns. Section \ref{credit} concludes and briefly describes the future development of the application to credit risk portfolio management.

\section{Distribution risk}\label{Sec:distrrisk}
This section introduces a measure of distribution model risk  for discrete models. According to \cite{breuer2016measuring} we are interested to measure the risk associated with uncertainty knowledge of a random variable distribution that belong to a class of plausible alternatives. We focus on discrete models, thus we consider the class $\mathcal{D}(\mathcal{Y}, \mu)$ of discrete probability mass functions (pmfs)   with support in  $\mathcal{Y}=\{y_1,\ldots, y_d\}$ and mean $\mu$, without loss of generality we assume  $y_1< y_2<\ldots y_d$. If $P\in \mathcal{D}(\mathcal{Y}, \mu)$ we write
 $p_j=P(Y=y_j)$, and $P=(p_1,\ldots, p_d)$.
Given a pmf in $\dym$ a  class of plausible alternatives is defined to be a subclass $\mathcal{P}$ of $ \dym$, we therefore do not measure uncertainty related to the mean value, that is given.
Define $\chi(\mathcal{Y})$ to be the set of random variables with support included in $\mathcal{Y}$.
The pmfs $P\in \mathcal{D}(\mathcal{Y}, \mu)$ are the  possible distributions of  $Y\in\chi(\mathcal{Y})$.
Given a random variable $Y\in \chi(\mathcal{Y})$, let  $\rho:\chi(\mathcal{Y})\rightarrow \RR$ be a law invariant functional, i.e. if $X$  and $Y$ have the same law -$X\sim Y$-  then $\rho(X)=\rho(Y)$ (\cite{bellini2021law}). We consider two different cases:  in the pricing application, $\rho(Y)$ is a payoff function  and,  in the portfolio application, it is a measure of risk.  We study how the choice of a wrong distribution for $Y$ in a plausible class of distributions $\mathcal{P}\subseteq \dym$ impact on $\rho(Y)$ and we build on the geometrical characterization of the class $\mathcal{D}(\mathcal{Y}, \mu)$.

We recall that a polytope (or more specifically a $d$-polytope)
is the convex hull of a finite set of points in $\RR^d$ called the extremal points of the polytope. We say that a set of $k$ points is affinely independent
if no one point can be expressed as a linear convex combination of the others. For example, three points are affinely independent if they are not on the same line, four
points are affinely independent if they are not on the same plane, and so on. The convex hull of $k+1$  affinely
independent points is called a simplex or $k$-simplex. For example, the line segment joining two points is a
1-simplex, the triangle defined by three points is a 2-simplex, and the tetrahedron defined by four points is a
3-simplex. A complete reference on computational geometry is \cite{de1997computational}.

In \cite{fontana2021model} the authors proved that the class $\mathcal{D}(\{0,\ldots, d\}, \mu)$ is a convex polytope.
The following Theorem \ref{discrgen} is a generalization of this result   and provides the  generators of $\mathcal{D}(\mathcal{Y}, \mu)$ in closed form.
\begin{theorem}\label{discrgen}
The class  $\mathcal{D}(\mathcal{Y}, \mu)$ is a convex polytope with extremal points:

\begin{equation}  \label{binuleq}
Q_{j_1,j_2}=\left\{
\begin{array}{cc}
q_{j_1}=\frac{y_{j_2}-\mu}{y_{j_2}-y_{j_1}} & y=y_{j_1} \\
q_{j_2}=1-q_{j_1} & y=y_{j_2} \\
0 & \text{otherwise}%
\end{array}
\right.,
\end{equation}
where ${j_2}>{j_1}\in \mathcal{Y}$, and
\begin{equation}\label{noArb0}
y_{j_1}< \mu< y_{j_2}
 \end{equation} i.e., any  measure $Q\in \mathcal{D}(\mathcal{Y}, \mu)$ has the representation
\begin{equation}\label{rep0}
{Q}=\sum_{i=1}^{n_{\mathcal{D}}}\lambda _{i}{Q}_{i},
\end{equation}
where ${Q}_{i}$ are the generators of $\mathcal{D}(\mathcal{Y}, \mu)$ and $n_{\mathcal{D}}$ is their number.

If $\mu=y_{j_1}$ or $\mu=y_{j_2}$ there is also

\begin{equation}  \label{binuleq}
Q_{\mu}=\left\{
\begin{array}{cc}
1&y=\mu\\
0 & \text{otherwise}%
\end{array}
\right..
\end{equation}

\end{theorem}

\begin{proof}

Let $P=(p_1,\ldots, p_d)$ be a discrete probability measure. We have $P\in \mathcal{D}(\mathcal{Y}, \mu)$ iff
\begin{equation}  \label{simp0}
\sum_{l=1}^d(y_j-\mu)p_j=0,
\end{equation}

with the conditions $p_j\geq0, \, j=0,\ldots,d$ and $\sum_{j=1}^dp_j=1$.
From the standard theory of linear equations we know that all the positive
solutions of \eqref{simp0} that sums up to one are elements of the convex polytope
\begin{equation*}  \label{cone}
\mathcal{P}=\{\boldsymbol{z}\in \mathbb{R}^{d}: \sum_{j=1}^d(y_j-\mu)z_j=0,\, I%
\boldsymbol{z}\geq 0, \sum_{j=1}^dz_j=1\},
\end{equation*}
where  $I$ is the $d\times d$ identity matrix.
The matrix  $B=[y_1-\mu,\ldots, y_d-\mu]$ is the row vector of the coefficients. Since $\rank B=1$ then an extremal point $Q$ of the polytope  has at most two non-zero components, say $q_{j_1}, q_{j_2}$. Therefore the extremal points can be found as the  positive solutions of  the reduced  linear system
\begin{equation}  \label{simp00}
\left\{ \begin{array}{c}
(y_{j_1}-\mu)q_{j_1}+(y_{j_2}-\mu)q_{j_2}=0,\\
q_{j_1}+q_{j_2}=1
\end{array}
\right.,
\end{equation}
where we make the non restrictive assumption $j_2>j_1$, therefore $y_{j_2}>y_{j_1}$.
The above system has a unique solution:

\begin{equation}  \label{binul}
Q_{j_1,j_2}=\left\{
\begin{array}{cc}
q_{j_1}=\frac{y_{j_2}-\mu}{y_{j_2}-y_{j_1}} \\
q_{j_2}=1-q_{j_1}\\
\end{array}
\right.,
\end{equation}
The solution in \eqref{binul}
is positive iff   $(y_{j_1}-\mu)(y_{j_2}-\mu)<0$ and this is equivalent to the condition
\eqref{noArb0}, i.e., $y_{j_1}< \mu< y_{j_2}$. Obviously if $\mu=y_{j_1}$ or $\mu=y_{j_2}$, $Q_{\mu}\in \mathcal{D}(\mathcal{Y}, \mu)$ and it has support on less than two points.
Thus the assert.
\end{proof}

We call the extremal points $Q^j:=Q_{j_1,j_2}$ of $\mathcal{D}(\mathcal{Y}, \mu)$ extremal pmfs or probability measures.

Let $Y\in \chi(\mathcal{Y})$, we measure the effect on $\rho(Y)$ of the uncertainty about  $P\in \mathcal{P}\subseteq \dym$ as in \cite{cont2006model}, by:

\begin{equation}\label{bounds1}
\mu_{\mathcal{P}}(Y)=\rho_U^{\mathcal{P}}(Y)-\rho_L^{\mathcal{P}}(Y),
\end{equation}
where $\rho_U^{\mathcal{P}}(Y)=sup_{p\in \mathcal{P}}\rho(Y)$ and $\rho_L^{\mathcal{P}}(Y)=inf_{p\in \mathcal{P}}\rho(Y)$.
We denote $\rho_U(Y)=\rho_U^{\dym}(Y)$ and  $\rho_L(Y)=\rho_L^{\dym}(Y)$.

If $\rho(Y)=E[\phi(Y)]$ for any measurable function $\phi$ the following proposition proves  that $\rho_U(Y)$ and $\rho_L(Y)$ are reached on the extremal generators. The proof is similar to the special case $\mathcal{Y}=\{0,\ldots, d\}$ in  \cite{fontana2021exchangeable}.

\begin{proposition}\label{bound0}
If $\rho(Y)=E[\phi(Y)]$ for a measurable function $\phi$  the bounds in \eqref{bounds1}   are attained at two   extremal points of $\dym$.
\end{proposition}\label{priceBound}
\begin{proof}

Let $Y\sim Q\in \mathcal{D}(\mathcal{Y}, \mu)$.

\begin{equation}
\begin{split}
E[\phi(Y)]&=\sum_{h=1}^{d}\phi(y_h)q_h=\sum_{h=1}^{d}\phi(y_h)\sum_{l=1}^{n_{\mathcal{D}}}\lambda_jq^j_h\\
&=\sum_{j=1}^{n_{\mathcal{D}}}\lambda_j\sum_{h=1}^{d}\phi(y_h)q^j_h=\sum_{j=1}^{n_{\mathcal{D}}}\lambda_jE[\phi(Y^j)],
\end{split}
\end{equation}
where $Q^j=(q^j_1,\ldots, q^j_d), j=1,\ldots {n_{\mathcal{D}}}$ are the extremal measures and $Y^j\sim Q^j$.
Therefore $E[\phi(Y)]$ is a point in a convex polytope generated by the expectations $E[\phi(Y^j)]$, thus the maximum and the minimum values of $E[\phi(X)]$   are reached on the extremal points.

\end{proof}

If $\phi$ is also convex we can explicitly find the corresponding extremal pmfs. This is important for our main application, where we measure the effect of model risk on convex derivatives  (e.g. options) prices.
The convex order is necessary to prove one of our main results.

\begin{definition}
Given two random variables $X$ and $Y$ with finite means, $X$ is said to be smaller
than $Y$ in the convex order (denoted $X\le_{cx}Y$) if
$$
E[\phi(X)]\leq E[\phi(Y)]
$$
for all real-valued convex functions $\phi$ for which the expectations exist.
\end{definition}
 The convex
order is a variability order, in fact it  is easy to verify that
$X\le_{cx}Y$ implies $E[X]=E[Y]$, and
${V}[X]\leq{{V}[Y]}$.
It can also be proved, see e.g. \cite{shaked2007stochastic}, that
\begin{center}
$X\leq_{cx}Y$ iff $E[X]=E[Y]$ and $E[(X-l)^+]\leq E[(Y-l)^+]$ for all $l\in \RR_+$,
\end{center}
where $x^+=max\{x,0\}$. See \cite{denuit1998optimal} for a complete overview on the convex order.

\begin{proposition}\label{cx}
Let $Y\in \mathcal{D}(\mathcal{Y}, \mu)$, then
\[
Y_L\leq_{cx}Y\leq_{cx} Y_U,\] where
$Y_L\sim Q_L$, $Y_U\sim Q_U$, $Q_L=Q_{y_M,y_m}$ with $y_M$ being the largest $y\in \mathcal{Y}$ smaller than $\mu$ and $y_m$ the smallest $y\in \mathcal{Y}$ bigger than $\mu$, and $Q_U=Q_{1,d}$.

\end{proposition}

\begin{proof}

Let $Y\sim Q\in \mathcal{D}(\mathcal{Y}, \mu)$. From Proposition \ref{bound0} for any measurable $\phi$
 the maximum and the minimum values of $E[\phi(X)]$   are reached on the extremal points.
It is therefore sufficient to prove that for any $\phi$ convex and
\[Y_L\leq_{cx}Y^j\leq_{cx} Y_U,\,\,\, \forall j\in \{1,\ldots, n_{\mathcal{D}}\}\]
where $Y^j\sim Q^j$ and $Q^j,  j\in \{1,\ldots, n_{\mathcal{D}}\}$ are the extremal measures.

Since $E[Y_L]=E[Y^j]=E[Y_U]$ by construction, it is sufficient to prove that
\[E[(Y_L-K)^+]\leq E[(Y^j-K)^+] \leq E[(Y_U-K)^+], \,\,\, \forall K\geq0.
\]
We prove that $E[(Y_L-K)^+]\leq E[(Y^j-K)^+]$.
Let $Y^j\in \mathcal{D}(\mathcal{Y}, \mu)$, $Y^j\sim Q^j$ with support on  $\{y_{j_1}, y_{j_2}\}$ and let $y_{l_1}<\mu<y_{l_2}$.
We consider three cases:

 Case 1: $K> j_m$.  We have $E[(Y_L-K)^+]=0$ and   $E[(Y^j-K)^+]>0$ for any $j_2>j_m$ and the assert is proved.

Case 2: $y_1<K<y_m$.  Since $y_{j_1}\leq y_M\leq y_m\leq y_{j_2}$, we have
\[
\begin{split}
 E[(Y-K)^+]=&(y_{j_1}-K)^+q_{j_1}+(y_{j_2}-K)^+q_{j_2}=(y_{j_2}-K)^+q_{j_2}=\\
&(y_{j_2}-K)\frac{\mu-y_{j_1}}{y_{j_2}-y_{j_1}}\geq(y_m-K)^+\frac{\mu-y_{j_1}}{y_{m}-y_{j_1}},
\end{split}
\]
where the last inequality follows from direct computations. Then
we observe that $\frac d{dy_{j_1}}(\frac{\mu-y_{j_1}}{y_{m}-y_{j_1}})<0$, thus $\frac{\mu-y_{j_1}}{y_{m}-y_{j_1}}$ is decreasing in $y_{j_1}$ and
\[
\begin{split}
 E[(Y-K)^+]\geq(y_m-K)^+\frac{\mu-y_{j_1}}{y_{m}-y_{j_1}}\geq (y_m-K)^+\frac{\mu-y_{M}}{y_{m}-y_{M}}=E[(Y_L-K)^+],
\end{split}
\]
 the assert is proved.

Case 3: $K<y_{j_M}$. If $K<y_{j_1}$ then
\[ E[(Y^j-K)^+]=\mu-K,\] for any $y_{j_1},y_{j_2}$ and the assert is proved.
If $ y_{j_1}\leq K\leq y_M$
\begin{equation}\label{caso3}
 E[(Y^j-K)^+]=(y_{j_2}-K)\frac{\mu-y_{j_1}}{y_{j_2}-y_{j_1}}\geq \mu-K,\end{equation} for any $y_{j_1},y_{j_2}$ and the assert is proved.

We now prove  that $E[(Y^j-K)^+] \leq E[(Y_U-K)^+]$.
We have
\[
\begin{split}
 E[(Y_U-K)^+]=
&(y_d-K)\frac{\mu-y_1}{y_{d}-y_{1}}.
\end{split}
\]

Case 1: $K>y_{j_2}$.  We have $E[(Y^j-K)^+]=0$ and the assert is proved.

Case 2: $y_{j_1}<K<y_{j_2}$. As in Case 2 above, we have
\[
\begin{split}
 E[(Y_U-K)^+]=
&(y_{d}-K)^+\frac{\mu-y_1}{y_{d}-y_{1}}\geq(y_{j_2}-K)^+\frac{\mu-y_{j_1}}{y_{j_2}-y_{j_1}}, \,\,\, \forall j_1,j_2.
\end{split}
\]

Case 3: $K<y_{j_1}$. If $k<y_{j_1}$ then

Straightforward computations give:
\[ E[(Y^j-K)^+]=\mu-K,\] for any $y_{j_1},y_{j_2}$,  the assert follows from Equation \ref{caso3} with $j_1=1$ and $j_2=d$.

\end{proof}
We thus have proved that if $\rho(X)=E[\phi(X)]$, where $\phi$ is a convex function and $\mathcal{P}=\dym$, than
\begin{equation}\label{globrisk}
\mu_{\mathcal{P}}(Y)=\rho(Y_U)-\rho(Y_L).
\end{equation}
The risk associated with the whole class is analytical, because $Q_L$ and $Q_U$ are analytical, nevertheless using the methodology developed in \cite{fontana2021exchangeable} we can find empirical bounds also on subclasses $\mathcal{P}\subset \dym$, as we do in our application to derivative pricing.

\section{\protect\bigskip Distribution risk in discrete  pricing models\label{Distr:price}}
We consider the effect of model risk on derivative pricing in incomplete discrete markets.
As in \cite{cont2006model} we consider uncertainty in the choice of the pricing model, specifically we study the risk arising from the choice of an equivalent martingale measure.
We consider a  market $\mathcal{M}$ with  a risky asset and a risk-free bond with $L$ states of the world. We  denote the price of
the stock at discrete times  $t_n, n = 0, 1, \ldots, N$ with $t_0 = 0$ by $\SS=\{S_n, n = 0, 1, \ldots,N\}$  and the risk-free bond  price by  $\BB=\{B_n=(1+R)^n, n = 0, 1, \ldots,N\}$, where $R$ is the single period risk-free rate.  The reference historical probability $P$ is represented as a vector $P=(p_1,\ldots, p_L)$, $\sum_{i=1}^Lp_i=1$, $p_i>0, i=1,2,\dots, L$.

After introducing the stock market price dynamics, using the results in Section \ref{Sec:distrrisk}, we find the generators of all the risk-neutral measures defined on the possible $L$ states of the world. We have an arbitrage free price for each measure equivalent to $P$ in the polytope. Usually, the risk-neutral measure selected to price is the minimal equivalent martingale measure (MEMM), that is unique given $P$. We measure the risk associated with this choice of the risk-neutral measure at three levels:
\begin{enumerate}
\item we find the analytical bounds for prices in the whole polytope, that define the maximal admissible  prices interval;
\item we find by simulations bounds for prices in the class of plausible risk-neutral measures, i.e. the  risk-neutral measures equivalent to $P$;
\item we measure the effect on prices of small perturbations of the MEMM.
\end{enumerate}

 We also provide an illustrative example on real data considering  an European style call option. For the simple case of European options we observe that increasing the number of steps the price interval is close to the no arbitrage interval  $I=(1/(1+R)^n\phi(S_0(1+R)^n), S_0)$ with $\phi$ convex, that is entirely spanned in most of the continuous incomplete models as proved in  \cite{eberlein1997range}.
\subsection{The model}

Assume that the stock price process $\SS=\{S_n,\, n\in \NN\}$ evolves randomly on an $L$ state lattice model (i.e.,
given the price of the stock at time $t_n, n = 1,\ldots, N-1$.
there are $L$ possible future
prices that it can take at time $t_{n+1}$). Suppose that $u$ and $d$ satisfy $u >d > 0$ then
a multinomial recombining lattice can be constructed by taking the $L$ possible future states for $S_{n+1}$ from $S_n$ as
\begin{equation}
S_{n+1} = u^{L-l}d^{l-1}S_n, \,l = 1,\ldots,L
\end{equation}
with probabilities $p_l>0 , l = 1,\ldots, L$, satisfying $\sum_{j=0}^Lp_j=1$. In this case, the
stock may achieve $n(L -1)+1$ possible prices at time $t = t_n, n = 0,\ldots, N$ given
by
\begin{equation}
S^{(k)}_n
= u^{n(L-1)+1-k}d^{k-1}S_0, \, k = 1, \ldots, n(L-1) + 1.
\end{equation}
Let now  $a_l:=u^{L-l}d^{l-1}$ so that we have
\begin{equation}
S_{n+1} = a_lS_n, \,l = 1,\ldots,L.
\end{equation}

\subsubsection{risk-neutral probabilities generators}

A strictly positive probability measure $Q=(q_1,\ldots, q_L)$ is said to be a risk-neutral probability measure for $\SS$ iff the discounted price process  is a $Q$-martingale, i.e. for any $n\in \NN$
\begin{equation}\label{MartCond0}
E^Q[S_{n+1}|S_n=S_n^{k}]=S_n^{(k)}(1+R).
\end{equation}
The no arbitrage condition, i.e. the existence of a risk-neutral probability measure, is $a_m\leq 1+R\leq a_M$, where   $a_m=\min_{i}a_i=d^{L-1}$ and $a_M=\max_ia_i=u^{L-1}$. We assume that  the market $\mathcal{M}$ is arbitrage free  and therefore an equivalent martingale measure exists. If there are more the two states of the world $L>2$ the market is not complete and the martingale measure is not unique. Theorem \ref{Qpolit} provides a geometrical  representation of  the class of all the risk-neutral probabilities as a convex polytope and analytically provides the set of extremal generators.

\begin{theorem}\label{Qpolit}
\label{binu} 

The risk-neutral probabilities $Q=(q_1,\ldots, q_L)$ are points in the convex polytope $$\mathcal{P}_Q=\{\boldsymbol{z}\in \mathbb{R}^{d}: \sum_{l=1}^d(a_l-(1+R))z_j=0,\, I%
\boldsymbol{z}\geq 0, \sum_{j=1}^dz_j=1\},$$
 whose generators are
%
\begin{equation}  \label{binuleq}
Q_{l_1,l_2}=\left\{
\begin{array}{cc}
q_{l_1}=\frac{(1+R)-a_{l_2}}{a_{l_1}-a_{l_2}} & y=a_{l_1}S_0 \\
q_{l_2}=1-q_{l_1} & y=a_{l_2}S_0 \\
0 & \text{otherwise}%
\end{array}
\right.,
\end{equation}
where $l_2>l_1>0$, and - no arbitrage condition -
\begin{equation}\label{noArb}
a_{l_2}\leq 1+R\leq a_{l_1}
 \end{equation} i.e., any equivalent martingale measure for $\SS$ has the representation
\begin{equation}\label{rep}
{Q}=\sum_{l\in{\l}}^{n_{L}}\lambda _{l}{Q}^{l},
\end{equation}
where $\mathcal{L}$ is the set of pairs $(l_1, l_2)$ that satisfy \eqref{noArb}, ${Q}^{l}$ are the generators of $\mathcal{P}_Q$ in \eqref{binuleq} and $n_L=\#\l$.
\end{theorem}

\begin{proof}
Let us introduce the random variable
$Z$ with support on $\{a_1,\ldots, a_L\}$ and distribution $P=(p_1,\ldots, p_L)$, i.e. $P(Z=a_l)=p_l$ and let it be  independent of $S_n$ , we can write
\begin{equation}
S_{n+1}=ZS_n.
\end{equation}
Equation \ref{MartCond0} is equivalent to
\[
E^Q[ZS_n|S_n=S_n^{k}]=S_n^{(k)}(1+R)
\]
and
\[
E^Q[Z]=(1+R)
\]
that is equivalent to
\begin{equation}\label{MartCond00}
\sum_{l=0}^La_{l}q_l=(1+R).
\end{equation}

We therefore have to find the solutions $Q=(q_1,\ldots, q_L)$ of \eqref{MartCond00}, i.e. we have  to find all the discrete distributions in $\mathcal{D}(\mathcal{A}, 1+R)$, where $\mathcal{A}=\{a_1,\ldots, a_L\}$.
Let $l_2>l_1>0$, we have
\begin{equation*}
\frac{a_{l_1}}{a_{l_2}}=(\frac{u}{d})^{l_2-l_1}>1,
\end{equation*}
and $a_{l_2}\leq a_{l_1}$. The assert follows from Proposition \ref{discrgen} observing that condition \eqref{noArb} is equivalent to condition \eqref{noArb0}.
\end{proof}
We call  $Q^l=Q_{l_1,l_2}, \,\,\, l=(l_1,l_2)\in \l$ a martingale measure extremal generator (MMEG).

\begin{remark}Theorem  \ref{Qpolit} implies that for $L=2$ the market is complete, i.e. the martingale measure is unique, and we find the Binomial model. For any $L$,
the generators of $\mathcal{P}_Q$ have support on two points. It follows from \eqref{binul} that they are the unique martingale measure of the complete Binomial model with $u'=a_{l_1}=u^{L-1_1}d^{l_1-1}$ and $d'=a_{l_2}=u^{L-l_2}d^{l_2-1}$. The conditions \eqref{noArb} are the no arbitrage conditions for each Binomial model. We call $B^l:=B_{l_1, l_2}, l\in \l$ the Binomial tree with the risk-neutral measure $Q^l, l\in \l$.
\end{remark}
The MMEGs  do not depend on the historical probability measure $P$ and in general they are not equivalent to the historical measure $P$. In fact, since $P$ and $Q$ are discrete measures, $Q$ is equivalent to $P$ ($Q\sim  P$)  iff they have the same support.
We introduce the family $\mathcal{Q}$ of plausible risk-neutral measures, as  the  measures in $\mathcal{P}_Q$ equivalent to $P$:
\begin{equation}
\mathcal{Q}=\{Q\in \mathcal{P}_Q: Q\sim P\}.
\end{equation}
\begin{remark}\label{Remark:abscont}
Notice that $\dym$ is a polytope in $\RR^{d-2}$, since  $P\in\dym$ if $P\in \RR^{d}$ and it satisfies two conditions -sum to one and mean. The pmfs without full support are included in lower dimension spaces. In fact if it exists $ j:$ $p_j=0$, then $P\in \RR^{d-1}$, thus the polytope is included in $\RR^{d-1-2}$ because of the two constraints - sum to one and mean that still remain. As a consequence uniform simulation from  $\dym$ or from $\mathcal{Q}$ gives the same results, in fact we have null probability to pick a pmf without full support. We also show that analytical bounds of prices and bounds found by simulation on $\mathcal{Q}$ are very close, for the same reason discussed above.
\end{remark}

The simplest example of incomplete lattice model is  the trinomial model.
We have three states of the world and therefore three states for the model: $\{a_1, a_2, a_3\}$. Without loss of generality assume $a_3>a_2>a_1$. The risk-neutral probabilities $Q=\{q_1, q_2, q_3\}$ are convex combinations of the risk-neutral probabilities associated with Binomial trees. The no arbitrage condition is $a_1<1+R<a_3$, thus we have two cases:
\begin{enumerate}
\item If $1+R<a_2$, from Theorem \ref{binu} the risk-neutral extremal points are two $Q_{12}$ and $Q_{13}$.
\item If $1+R>a_2$, from Theorem \ref{binu} the risk-neutral extremal points are two $Q_{23}$ and $Q_{13}$.
\end{enumerate}

\begin{example}\label{trinomial}
We propose a trivial but explicative example.
Let us consider a one step trinomial model, i.e. $L=3$, with $u=1.2,\,  d=0.8,\, R=2\%,\,S_0=100$.
The support of $S_1$ is $\{144, 96, 64 \}$ and the no arbitrage condition in satisfied. The generators of the risk-neutral probabilities are $Q_{{1},3}=\{0.475,
    0.525\}$ and $Q_{1,2}=\{0.125,
    0.875\}$. A trinomial recombining  tree with the two Binomial generators is shown in Figure \ref{TT}.

\begin{figure}[h!]
\caption{Trinomial tree. The green line indicates the Binomial tree $1,3$; the red line indicates the Binomial tree $1,2$. The orange line is the overlapping of the two Binomial trees.}\centering
\includegraphics[width=0.5\linewidth]{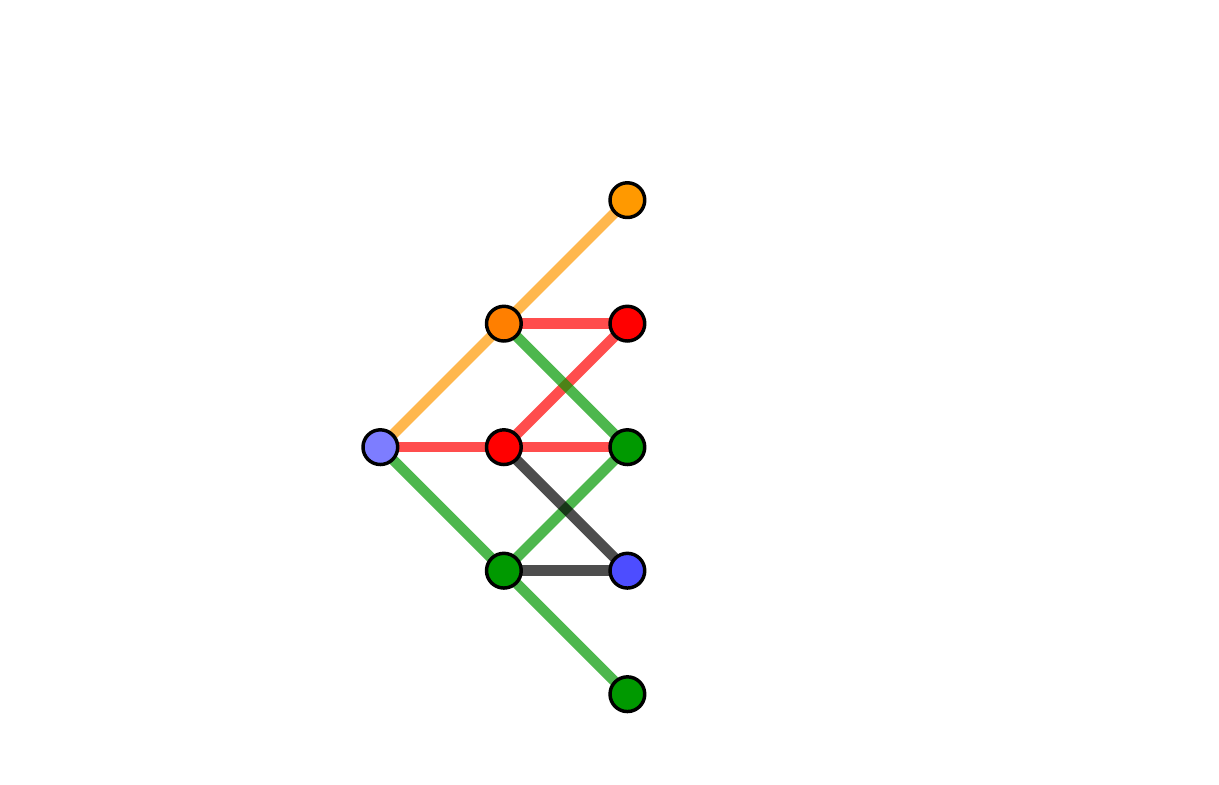}
\end{figure}
\label{TT}
\end{example}

\subsubsection{Minimal entropy martingale measure}

Derivatives evaluation in incomplete markets requires  a criterion to
choose a suitable risk-neutral pricing measure.
A standard criterion
 is based
on the  relative entropy with respect to the historical reference measure $P$.
\begin{definition}
Let Q and P be probability measures on a finite probability space $\mathcal{Y}$.  The relative
entropy of Q with respect to a probability measure P is a number defined as
\begin{equation}
I(Q, P)=\sum_{i=1}^Lq_i\ln(\frac{q_i}{p_i})
\end{equation}
with the assumption that  $0ln(0)=0$. If $Q$ and $P$ do not have the same support we define $I(Q, P)=+\infty$.
\end{definition}
\begin{remark}
The historical measure $P$ measure is assumed to have full support, thus the extremal pmfs are not equivalent to $P$ and have relative entropy $+\infty$.
\end{remark}

Among all the risk-neutral measures we select a measure that minimizes the relative entropy with respect to $P$. Formally, a probability measure $\tilde{Q}$ is called the Minimal entropy martingale measure (MEMM) if it satisfies
\begin{equation}
I(\tilde{Q},P)=\min_{Q\in \mathcal{M}}I(Q,P).
\end{equation}
To find a minimal entropy measure we have to solve the problem
\begin{equation}
\begin{cases}
&\min_{Q\in \mathcal{M}}I(Q,P)\\
&\sum_{l=1}^Lq_l=1\\
&\sum_{l=1}^La_{l}q_l=1+R.
\end{cases}
\end{equation}
In the multinomial model, under the no arbitrage condition, the MEMM $\tilde{Q}=(\tilde{q}_1,\ldots, \tilde{q}_L)$ does exist and is unique (\cite{frittelli1995minimal}). We can determine the MEMM using the
method of Lagrangian multipliers \cite{ssebugenyi2013minimal} and we find:
\begin{equation}\label{MMq}
\tilde{q}_i=\frac{p_ie^{-\tau a_i}}{\sum_{j=1}^Np_je^{-\tau a_j}}, \,\,\, i=1,\ldots, L,
\end{equation}
where $\tau$ is the solution of
\begin{equation}
\sum_{j=1}^La_jp_je^{-\tau a_j}=(1+R)\sum_{j=1}^Lp_je^{-\tau a_j}.
\end{equation}

The minimal relative entropy  is equivalent to the real world $P$ and therefore cannot be  reached on the extremal points of the polytope, thus the MEMM  is not a MMEG. Nevertheless, using the procedure in \cite{fontana2021exchangeable} we can simulate the distribution of the relative entropy across the class $\mathcal{Q}$.  We measure the risk of selecting the MEMM measure by considering the effect on prices consequent to a perturbation  of $\tilde{Q}$. We then consider the following class of distributions in $\mathcal{Q}$:

\begin{equation}\label{gamma}
\Gamma(\tilde{Q}, \epsilon)=\{Q\in \mathcal{Q}:\,\,\,  I(\tilde{Q},Q)\leq \epsilon\}\,\,\, \text{for} \,\,\,\epsilon>0
\end{equation}
and then we study the price distribution and their range across this class.

\setcounter{example}{0}
\begin{example}[\textbf{continued}]
 Let the historical probability measure be $P=(0.3\,\,\,0.4\,\,\,0.3)$, the corresponding MEMM is $\tilde{Q}=(0.2866\,\,\, 0.3974\,\,\,0.3161)$. It holds $\tilde{Q}=\tilde{\lambda}Q_L+(1-\tilde{\lambda})Q_U$ with $\tilde{\lambda}=0.4541$.
   Table \ref{tabI} provides the minimal  relative entropy $I(\tilde{Q}, P)$ and the relative entropies of $Q_L$ and $Q_U$.
\begin{table}[h!]
\begin{center}\caption{}
\begin{tabular}{r |r    } \label{tabI}
Measure&Relative entropy\\
\hline
$Q_L$&$+\infty$\\
 $\tilde{Q}$&0.0073\\
$Q_U$&$+\infty$\\
 \end{tabular}
\end{center}
\end{table}
 Figure \ref{crossE} shows the simulated  entropy distribution across the class $\mathcal{Q}$ of risk-neutral martingale measures equivalent to $P$.

\begin{figure}[h!]
\caption{Relative entropy distribution across $\mathcal{Q}$.}\centering
\includegraphics[width=0.3\linewidth]{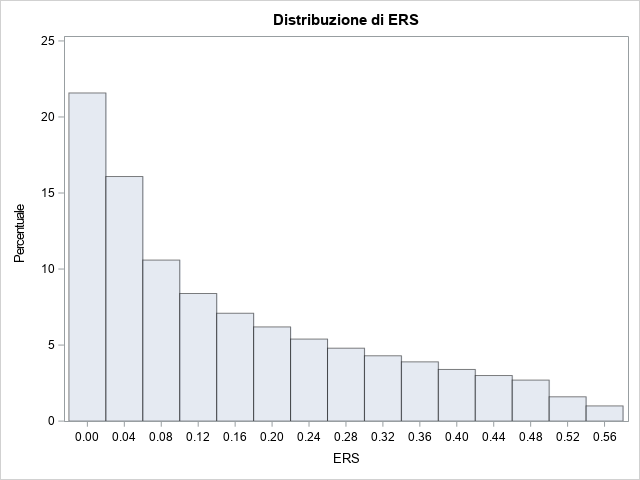}
\end{figure}
\label{crossE}
%
%
%
%
%
%
%
\end{example}

The next Section provides the analytical bounds, and therefore an analytical interval, for prices of convex derivatives in the whole class $\dym$.

\subsection{Analytical  price bounds on $\dym$}\label{AnalBounds}

Let $\phi$ be a positive  function. Suppose there is a contingent claim which pays off an amount $X=\phi(S_{N})$ at time $t_N$. Let $\pi(t, X)$ be the derivative value at time $t$. A no arbitrage  price of $X$ is given by
\begin{equation}
\pi(0, X)=\frac{1}{(1+R)^N}E^{Q^N}[{X}], \,\,\, Q\in \mathcal{P}_Q
\end{equation}
We can derive bounds for $\pi(0, X)$ as follows:
\begin{equation}\label{bounds}
\inf_Q\frac{1}{(1+R)^N}E^{Q^N}[{X}]\leq\pi(0, X)\leq\sup_Q\frac{1}{(1+R)^N}E^{Q^N}[{X}],
\end{equation}
where $Q^N$ is the risk-neutral distribution of $S_N$ corresponding to $Q$. Without loss of generality we assume $R=0$, $S_0=1$ and  \eqref
{MartCond0}
becomes
\begin{equation}\label{MartCond1}
E^Q[S_{n+1}|S_n=S_n^{k}]=S_n^{(k)},
\end{equation}
We define $\pi_L(0,X)=\inf_{Q^N}E^{Q^N}[{X}]$ and  $\pi_U(0,X)=\sup_{Q^N}E^{Q^N}[{X}]$.
Following \cite{cont2006model} we define model uncertainty $\mu_{\mathcal{P}_Q}=\pi_U(0,X)-\pi_L(0,X)$.
The following is a corollary to Proposition \ref{bound0}.
\begin{corollary}\label{bound}
In the one-period model the bounds in \eqref{bounds}  $\inf_QE^{Q}[{X}]$ and $\sup_QE^{Q}[{X}]$ are attained at two   extremal points of $\dy$.
\end{corollary}\label{priceBound}
%

Let $\pi_U(0,X_1)=E^{Q_U}[\phi(S_1)]$  and  $\pi_L(0,X_1)=E^{Q_L}[\phi(S_1)]$, they are the maximum and minimum possible prices and therefore any other price, corresponding to a martingale measure $Q$  can then be expressed by:
\begin{equation}\label{Lprice}\pi_Q(0,X)=\lambda\pi_L(0,X)+(1-\lambda)\pi_U(0,X), \,\,\, \lambda\in [0,1].
\end{equation}
The two extremal measures are not equivalent to $P$ and they generates two Binomial trees.
 Since we are able to analytically find all the generators, we can also find the maximum and minimum price for each contingent claim as defined above.

 Theorem \ref{Q:boundCall} is our main results and it provides analytical bounds for prices of convex derivatives in a multi-period lattice model.


\begin{theorem}\label{Q:boundCall}

Let $X=\phi(S_n), \, n\geq1$ and let $\phi$ be a convex function. We have:
\[
\pi_L(0,X)\leq \pi(0,X) \leq \pi_U(0,X)\]
where:
\[\pi_L(0,X)=E^{Q^n_L}[\phi(S_{n})], \]
with $Q^n_L=Q^n_{l_M, l_m}$, $l_M$ is the maximum index such that  $a_{l_M}\leq 1$ and $l_m$ is the minimum index such that  $a_{l_m}\geq1$,
and
\[\pi_U(0,X)=E^{Q^n_U}[\phi(S_n)], \]
with $Q^n_U=Q^n_{1, L}$.
\end{theorem}
\begin{proof}
%

Let $Y$ be the discrete random variables with support $\{0,\ldots, L-1\}$ and pmf $Q_Y$  representing the number of jumps down of $\SS$ in a unit time, i.e. $Q_Y(Y=k)=Q(S_{1}=u^{L-k-1}d^k)$, and let $Y_i\sim Y$ be independent for $i=1,\ldots, N$. We have:

\[
S_{n}=u^{n(L-1)-\sum_{i=1}^nY_i}d^{\sum_{i=1}^nY_i}, n=1,\ldots, N.
\]
We have to prove that
\[
E^{Q^{n}_L}[\phi(S_{n})]\leq E^{Q^{n}}[\phi(S_{n})]\leq  E^{Q^{n}_U}[\phi(S_{n})],
\]
we prove only
\[
 E^{Q^{n}}[\phi(S_{n})]\leq  E^{Q^{n}_U}[\phi(S_{n})],
\]
 because  proof of the other inequality  is similar.
We first prove the result for $n=1$.
For any  convex function $\phi$,
\[E^Q[\phi(S_1)]=E^{Q_Y}[\psi(Y_1)], \]
where $\psi(x)=\phi(u^{(L-1)-x}d^x)$ is a convex function since it is the composition of two convex functions.  We have  $Y_1\in \dy$ (since $R=0$), and by Proposition  \ref{cx}
\[
Y_L\leq_{cx}Y_1\leq_{cx}Y_U, \,\,\,\forall n\geq1\] where
$Y_L\sim Q_{Y,L}\in \dy$, $Y_U\sim Q_{Y,U}\in \dy$, $Q_{Y,L}=Q_{y_M,y_m}$ with $y_M$ being the largest $y\in \mathcal{Y}$ smaller than $\mu=1$ and $y_m$ the smallest $y\in \mathcal{Y}$ larger than $\mu=1$, and $Q_{Y,U}=Q_{1,{L-1}}$.
Let $S_1^L=u^{(L-1)-Y_L}d^{Y_L}$; $S_1^U= u^{(L-1)-Y_U}d^{Y_U}$, we have $S_1^L\sim Q_L$ and $S_1^U\sim Q_U$ by construction and:

\[
E[\phi(S_1)]=E[\psi(Y_1)]\leq E[\psi(Y_U)]=E[\phi(S_1^U)],
\]
where the inequality holds observing that if $\phi$ is convex $\psi$ is convex.
Since $Y_i$ are i.i.d. by  the closure properties of the convex order (see \cite{shaked2007stochastic}) we have

\begin{equation}\label{cxX}
\sum_{i=1}^nY^i_L\leq_{cx}\sum_{i=1}^n Y_i\leq_{cx}\sum_{i=1}^nY^i_U,\end{equation} where $Y_L^i\sim Y_L$ and are independent, $Y_U^i\sim Y_U$ and are independent. We finally have:

\[
E^{Q^N}[\phi(S_n)]=E^{Q^N_{x}}[\psi(\sum_{i=1}^nY_i)]\leq E^{Q^N_{Y,U}}[\psi(\sum_{i=1}^nY_U^i)]=E^{Q^N_U}[\phi(S_n^U)],
\]
where the first and the last equality are by construction and the inequality follows  from \eqref{cxX}.
\end{proof}

Corollary \ref{cor} comes observing that  $f(y)= (K-y)^+$ [$f(y)=(y-K)^+$] is a convex function.

\begin{corollary}\label{cor}
Let $X=(K-S_N)^+$ [$X=(S_N-K)^+$] be an European put [call] option. Then, for any choice of the strike price $K$ the risk-neutral probability that gives the lower price is
$Q^N_I=Q^N_{l_M, l_m}$, where $l_M$ is the maximum index such that  $a_{l_M}\leq 1$ and $l_m$ is the minimum index such that  $a_{l_m}\geq1$, and the risk-neutral probability that gives the higher price is $Q^N_U=Q^N_{1, L}$.
\end{corollary}

Increasing the number of steps the bounds found in Corollary \ref{cor} goes to the analytical no arbitrage bounds  $1/(1+R)^n\phi(S_0(1+R)^n)$ and $S_0$. This is not surprising, since \cite{eberlein1997range} proved that most of the incomplete models span the whole range of admissible prices by changing the risk-neutral measure and multinomial models are approximating models. For this reason, in our example we will study empirical bounds on a subclasses of $\dym$, to show how the knowledge of the class $\dym$ allow us to measure the distribution risk associated with  perturbation of the risk-neutral measure chosen.

The analytical bounds are more interesting for American options, that are also convex derivatives.  American options have the same payoff of European options at maturity $N$, but they can be exercised at any time before maturity, i.e. any $n$, $ 0\leq n\leq N$, with payoff $\phi(S_n)$.
To  price  an American option one must account for the possible exercise policies.
In the simple case without dividends,
the value of the American call is equal to that of the European call and  holding to
expiration is optimal. On the contrary, the optimal exercise policy for American put options  is not to hold until expiration no matter what. Therefore, the price of a put option is higher than the price of the corresponding European put.
Proposition \ref{Q:boundPutA} also follows from convexity of the payoff, see \cite{ekstrom2004properties}, and gives analytical bounds for the price of an American put option.

\begin{proposition}\label{Q:boundPutA}
Let us consider an American put option with maturity $N$ and final payoff $X=(K-S^N)^+$. Then, for any choice of the strike price $K$, the risk-neutral probability that gives the lower no arbitrage price is
$Q_L=Q_{l_M, l_m}$, where $l_M$ is the maximum index such that  $a_{l_M}\leq 1$ and $l_m$ is the minimum index such that  $a_{l_m}\geq1$, and the risk-neutral probability that gives the higher price is $Q_U=Q_{1, L}$.
\end{proposition}

\begin{proof}


Let $E^Q_k[X]=E^Q[X|S_n=S_n^{(k)}]$. Consider the multinomial tree starting in node $(n,k)$, for any $(n,k)$.
Since
\begin{equation}\label{MartCond}
E^Q[S_{n+m+1}|S_{n+m}=S_{n+m}^{(k)}]=S_{n+m}^{(k)}.
\end{equation}
 is equivalent to

\begin{equation}\label{MartCond2}
\sum_{l=0}^La_{l}q_l=1,
\end{equation}
by Corollary \ref{cor} the continuation value $\pi(n, k, X)=E[\phi(K-S^N)^+|S_n=S_n^k]$  of  the put option satisfies
\begin{equation*}\label{boundsP}
\pi_L(n, k,X)\leq\pi(n,k , X)\leq\pi_U(n, k, X).
\end{equation*}
Thus
\begin{small}
\begin{equation*}
\max\{\pi_L(n, k, X), (S_n^k-K)^+\}\leq \max\{\pi(n, k, X), (S_n^k-K)^+\}\leq \max\{\pi_U(n, k, X), (S_n^k-K)^+\},
\end{equation*}
for any $(n,k)$ and the assert is proved.
\end{small}

\end{proof}

\setcounter{example}{0}
\begin{example}[\textbf{continued}]
Let us consider  a call option  on $S_t$.
The price of a call option with strike $K=100$ and maturity $T=t_N=N=1$ is bounded by the prices obtained with  the Binomial trees $l_1, l_3$, that is  $c_{13}=20.490$, and $l_2, l_3$, that  is  $c_{23}=5.3922$. All the no arbitrage prices  are convex combinations of $c_{13}$ and $c_{23}$ that are the upper and lower bounds.The MEMM price is $\tilde{c}=16.637$. Figure \ref{crossP}, left side, shows the distribution of prices across $\mathcal{P}_Q$.

\begin{figure}[h!]
\caption{Price distribution for $L=3$,  $N=1$ (left) and price distribution for $L=3$,  $N=10$ (right) .}\centering
\begin{subfigure}[b]{0.3\textwidth}
\includegraphics[width=\textwidth]{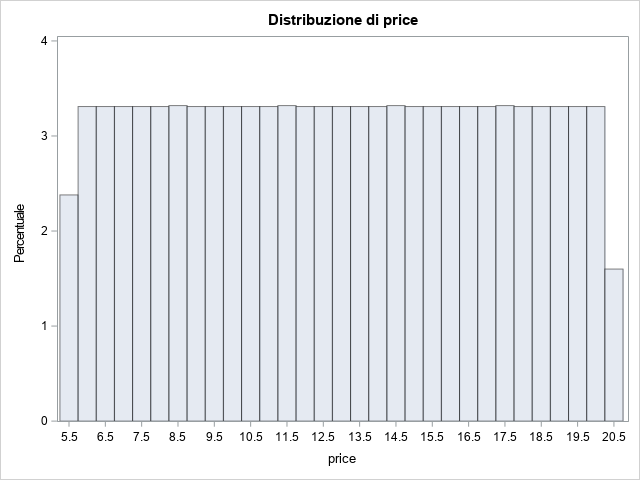}
\end{subfigure}
\begin{subfigure}[b]{0.3\textwidth}
\includegraphics[width=\textwidth]{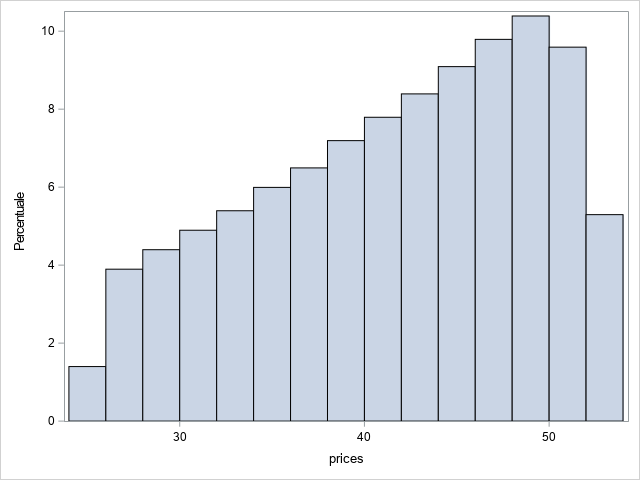}
\end{subfigure}
\label{crossP}
\end{figure}

%
Consider now ten steps, $T=N=10$. The ten periods trinomial model has the above risk-neutral one step probabilities, i.d. the risk-neutral probabilities in the polytope generated by $Q_{{1},3}$ and $Q_{{2},3}$. According to proposition \ref{Q:boundCall}  the risk-neutral prices are convex combinations of the two Binomial prices $c^{10}_{13}=53.4648$ and $c^{10}_{23}=25.2834$, that are the lower and upper bounds for a risk-neutral price. The MEMM price is $\tilde{c}^{10}=43.7139$. Figure \ref{crossP} (right side) shows the price distribution across the class.

%
%

\end{example}
We conclude this section with  an example with $L=4$ states of the world. This is the simplest case  with more than two generators, for this reason we keep Example \ref{EL4} also to illustrate the distribution risk across $\mathcal{Q}$ and the effect on prices of small perturbation of the MEMM $\tilde{Q}$.

\setcounter{example}{1}
\begin{example}
\label{EL4}
Let us keep $u=1.2,\,  d=0.8\, R=2\%,\,S_0=100, \, K=100$ as in Example \ref{trinomial}. Let us consider a one step multinomial model with $L=4$. The no arbitrage condition in satisfied.
The support of $S_1$ is $\{172.8; 115.2;76.8;51.2 \}$. The generating Binomial trees, their risk-neutral probabilities and the corresponding option prices are in Table \ref{tabray}.

\begin{table}[h!]
\begin{center}\caption{MMEGs and corresponding one-period call option prices with  $u=1.2,\,  d=0.8\, R=2\%,\,S_0=100, \, K=100$. We have $l_1<l_2$ and $a_{l_1}>a_{l_2}$.}
\begin{tabular}{r |r | r  |r } \label{tabray}
Binomial tree&$q_{l_2}$&$q_{l_1}$&Price\\
\hline
$c_{13}$& 0.2062
  & 0.7938&	11.8284\\
$c_{14}$&0.5822&
  0.4178&29.8168 \\
$c_{23}$& 0.3437&
   0.6563 & 9.7794  \\
$c_{24}$& 0.7375&
    0.2625&18.7353\\
\end{tabular}
\end{center}
\end{table}
The MMEGs $Q_L$ and $Q_U$, the MEMM $\tilde{Q}$ for a reference historical measure $P=(0.09\,\,\,  0.40\,\,\, 0.47 \,\,\,0.04)$ with support $S_1$  and the corresponding relative entropies are in Table \ref{tabRE}.
\begin{table}[h!]
\begin{center}\caption{Prices bounds and their generators,  MEMM and corresponding price. Probabilities entries corresponds to the points of the support in decreasing order.}
\begin{tabular}{r |r | r  |r } \label{tabRE}
Measure&probabilities&Relative entropy&Price\\
\hline
$Q_L$&$(0\,\,\,0.3437,\,\,\,0.6562\,\,\,0)$&$+\infty$&9.7794\\
 $Q_U$&$(0.5822\,\,\,0\,\,\,0\,\,\,0.4178)$&$+\infty$&20.8168\\
 $\tilde{Q}$&($0.0684\,\,\,0.5280\,\,\,0.3402\,\,\,0.0636)$&0.0207&12.7391\\
 \end{tabular}
\end{center}
\end{table}
%
%
%
%
%
%
%
%
%
%
%
%
%
%
%
%
The maximum and minimum prices are in correspondence of the Binomial trees $c_{14}$ and $c_{23}$, respectively.
\end{example}

In Theorem \ref{Q:boundCall} and Proposition \ref{Q:boundPutA} analytical bounds are found
using two risk-neutral measures that are not equivalent to $P$, since they are extremal and have support on two points and they generate two Binomial trees.  Section \ref{priceQ} finds the bounds in $\mathcal{Q}$ and empirically shows that the bounds found are close to the analytical.

\subsection{Price distribution and bounds on $\mathcal{Q}$}\label{priceQ}

 In this section we look for the bounds in $\mathcal{Q}$ and for the price distribution across $\mathcal{Q}$, the class of measures in $\dym$ equivalent to $P$. We  proceed by simulations, building on the simulation algorithm developed in \cite{fontana2021exchangeable}. We simulate uniformly from $\dym$ and throw out the measures  $Q$ not equivalent to $P$. According to Remark \ref{Remark:abscont} all the $Q$ randomly extracted from $\dym$ have full support, i.e. they belong to  $\mathcal{Q}$.

\setcounter{example}{1}
\begin{example}[\textbf{continued}]
\label{EL5}
We now consider the class  $\mathcal{Q}\subset\dym$, we recall that: $L=4,\,u=1.2,\,  d=0.8\, R=2\%,\,S_0=100, \, K=100$. We first consider one step, i.e. $N=1$. Figure \ref{cross1} (left side) shows the distribution of the relative entropy across the class and figure \ref{cross1} (right side) shows the distribution of prices across the class $\mathcal{Q}$.
\begin{figure}[h!]
\caption{Entropy (left) and price (right) distribution for $L=4$ and  $N=1$.}\centering
\begin{subfigure}[b]{0.3\textwidth}
\includegraphics[width=\textwidth]{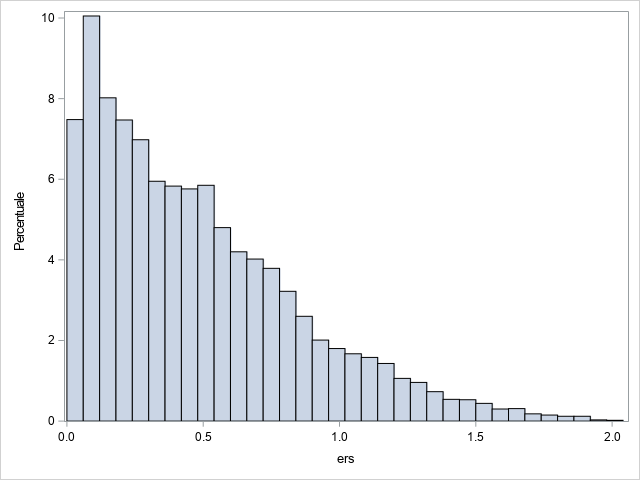}
\end{subfigure}
\begin{subfigure}[b]{0.3\textwidth}
\includegraphics[width=\textwidth]{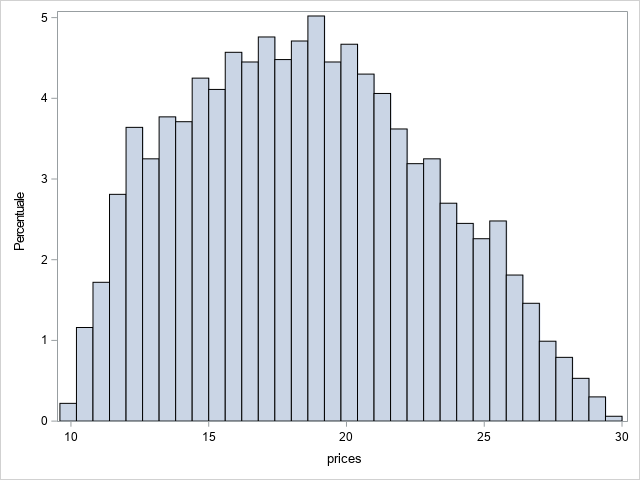}
\end{subfigure}
\label{cross1}
\end{figure}
Figure \ref{cP-10} shows the price distribution across the class in a ten step multinomial model, i.e. $N=10$.
%

%
\begin{figure}[h!]
\caption{Price  distribution for $L=4$ and  $N=10$.}\centering
\includegraphics[width=0.3\linewidth]{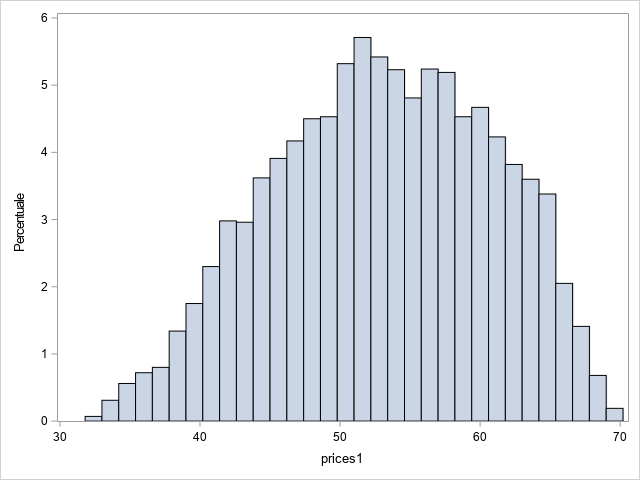}
\end{figure}\label{cP-10}
Table \ref{tabN10} gives the MEMM price and the analytical bounds with $N=10$ and, as one can see the simulated bounds across the class of  equivalent measures $\mathcal{Q}$ are close to the analytical bounds found on the MMEGs.
\begin{table}[h!]
\begin{center}\caption{Analytical price bounds and MEMM call option  with  $L=4,\, u=1.2,\,  d=0.8\, R=2\%,\,S_0=100, \, K=100,\,N=10$. }
\begin{tabular}{r |r  } \label{tabN10}
Martingale measure&Price\\
\hline
$Q_{U}$& 70.0699\\
$Q_{L}$& 31.9443 \\
$\tilde{Q}$&41.3017\\
\end{tabular}
\end{center}
\end{table}

In this case the analytical bounds given by the no arbitrage condition ({\bf check}) are $1/(1+R)^n(max(S_0(1+R)^n-K,0))= 17.9652$ and $S_0=100$.

\end{example}
\subsection{MEMM perturbation: effect on prices}
This section develops an example where we measure  the risk associated with small perturbation of the MEMM. Also in this case we do not have analytical bounds and we proceed by simulations. In the following continuation of Example \ref{EL4}, we  simulate the distribution of prices across
the set $
\Gamma(\tilde{Q}, \epsilon)$ in \eqref{gamma}, following the algorithm in \cite{fontana2021exchangeable}, reported in Appendix \ref{US} and find the simulated bounds for prices.

\setcounter{example}{1}
\begin{example}[\textbf{continued}]
\label{EL6}

Figure \ref{palla} shows the polytope $\mathcal{D}(\mathcal{Y}, \mu)$, dots represent uniformly simulated $q\in \mathcal{Q}$ outside the set $\Gamma(\tilde{Q}, 0.05)$, that is the white set inside the polytope. The point inside the white  ball $\Gamma(\tilde{Q}, 0.05)$ is $\tilde{Q}$.

\begin{figure}[h!]
\caption{Polytope $\mathcal{P}_Q$, sampled $q$ in $\mathcal{Q}$,  set $
\Gamma(P, 0.05)$ -white ball- for $L=4$.}\label{palla}\centering
\includegraphics[width=0.4\linewidth]{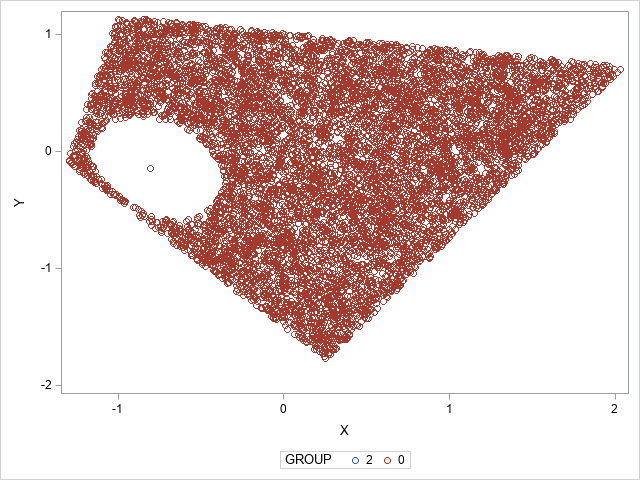}
\end{figure}

 Table \ref{tabPriceball} exibits (simulated)  price bounds on $\Gamma(\tilde{Q}, 0.05)$ obtained with a one-step and $N=10$ step binomial tree.

\begin{table}[h!]
\begin{center}\caption{Simulated price bounds on  $\Gamma(\tilde{Q}, 0.05)$ and MEMM call option  with  $L=4,\, u=1.2,\,  d=0.8\, R=2\%,\,S_0=100, \, K=100,\,N=10$. }
\begin{tabular}{r |r | r   } \label{tabPriceball}
&Price N=1& Price N=10\\
\hline
max& 15.2084&47.3172\\
min& 10.6699&35.2663\\
$\tilde{Q}$&12.7391&41.3017\\
\end{tabular}
\end{center}
\end{table}

Finally, Figure \ref{crosspalla}  shows the distribution of prices across the class $\Gamma(\tilde{Q}, 0.05)$ for $N=1$ and $N=10$ steps.

\begin{figure}[h!]
\caption{Price distributions across  $\Gamma(\tilde{Q}, 0.05)$ for $N=1$ (left) and $N=10$ (right) steps ($L=4$).}\label{crosspalla}\centering
\begin{subfigure}[b]{0.3\textwidth}
\includegraphics[width=\textwidth]{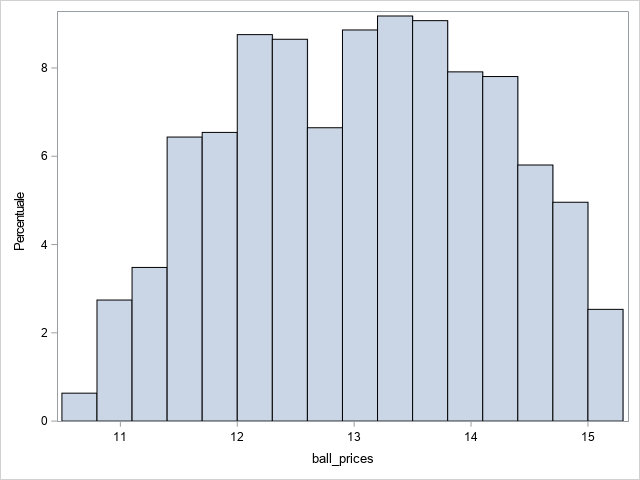}
\end{subfigure}
\begin{subfigure}[b]{0.3\textwidth}
\includegraphics[width=\textwidth]{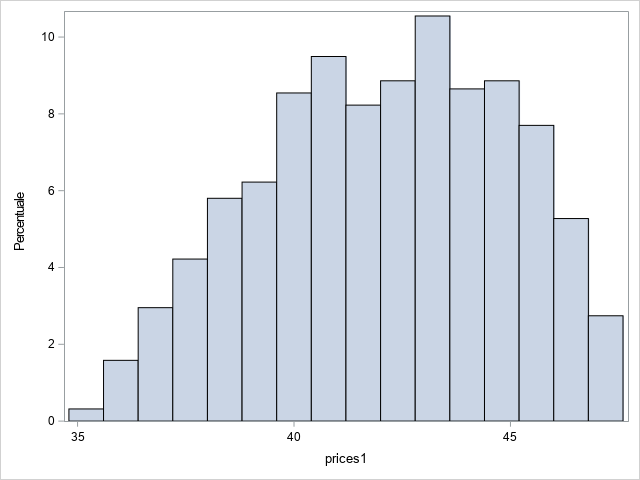}
\end{subfigure}
\end{figure}

As one can see comparing Table \ref{tabPriceball} and \ref{tabN10}, the range of prices spanned on a perturbation of the measure $\mathcal{Q}$ is significantly smaller than the whole range of prices.

\end{example}

\section{Example on real data}\label{real}


As a first illustrative application we study the distribution of risk-neutral European call prices and their bounds for a lattice with realistic parameters, i.e. parameters are chosen to match the moments of real asset returns.
Empirical study of financial log returns data shows the presence of significant degree
of skewness and excess kurtosis.
To incorporate such information, higher
order lattices are used. We follow \cite{yamada2004properties} to calibrate the pentanomial lattice parameters on the empirical moments of  observed log-returns.
Let us consider  i.i.d.  daily log returns
\begin{equation}
R_t=log(\frac{S_{t+1}}{S_t}), \,\,\, t=0,\ldots T-1,
\end{equation}
To incorporate mean, variance skewness and kurtosis in the price distribution we calibrate
the probabilities $P=(p_1,\ldots, p_5)$ by solving the
five linear equations (2.11)  in \cite{yamada2004properties}. Their expression and the corresponding jump amplitudes $(a_1,\ldots, a_5)$ are in equation (16) in \cite{ssebugenyi2013minimal}. We report them in Appendix \ref{Pentanomial}. Notice that jump amplitudes depend on the kurtosis and not on the skewness of log-returns, while $p$ depends on both. This is because the jump amplitudes are assumed to be symmetric by construction.
Applying Theorem \ref{Qpolit} we can find the generators of the polytope of the risk-neutral probabilities corresponding to the pentanomial lattice and consequently we can find the minimum and maximum option prices, the MEMM measure and the corresponding price. We  proceed as follows:
\begin{enumerate}
\item  We find the empirical mean, variance, skewness and kurtosis of daily returns: $\hat{\mu}$, $\hat{\sigma}$, $\hat{s}$ and $\hat{k}$;
\item we find the  pentanomial lattice jump amplitudes and historical probabilities using Equations \eqref{Ja} and \eqref{Jp}, respectively;
\item we find the MMEGs and   the MEMM;
\item we consider European options and we find the minimal entropy price, bounds for the risk-neutral prices of call and put options in the whole class of martingale measures and in a class of perturbation of the MEMM and - for $N=4$ years- the call price distribution across the class;

\end{enumerate}
We consider daily logreturns of Moncler  S.p.A. (MONC.MI) on FTSE MIP Market
Index from November 10, 2017 to November 10, 2022, with a total of 1270  daily  observations.
The empirical  daily mean, variance, skewness and excess kurtosis are reported in Table \ref{bla}.
\begin{table}[h!]
\begin{center}\caption{}\label{bla}
\begin{tabular}{r |r | r  |r }
$\hat{\mu}$& $\hat{\sigma}$& $\hat{s}$ & $\hat{k}$\\
\hline
0.0006& 0.0005& 0.1019& 4.4305\\
\end{tabular}
\end{center}
\end{table}
We have $S_0=21.5381$.
 We choose the annual rate $R=0.05$, we consider four maturities $T=1; 2;3;4;5$ years and K=22.
 Since the only purpose  of this application is to show how the methodology applies with realistic numbers, we construct a simple toy example.
  Each step is one year, so we have  $N=T$ steps. For $T=1$ we have the one-period model and we  use annualized parameters to find the jump amplitudes and the historical probability. Then we show the results we can obtain in a $N$-period model simply considering $T=2,3,4$ years.
We provide as a benchmarks in Table \ref{BS} the  Black and Scholes put and call prices.
\begin{table}[h!]
\begin{center}\caption{Black and Scholes put and call prices}\label{BS}
\begin{tabular}{c|c|c}
&Call Price&Put Price\\
\hline
$T=1$& 3.3025& 2.7031\\
$T=2$& 4.9826&3.3509\\
$T=3$ &6.3034&3.7009\\
$T=4$& 7.4272& 3.9011\\
\end{tabular}
\end{center}
\end{table}
The  jump amplitudes obtained from the annualized  empirical mean, variance, skewness and excess kurtosis  are in Table \ref{amp}.
\begin{table}[h!]
\begin{center}\caption{}\label{amp}
\begin{tabular}{r |r | r  |r |r}
$a_1$&$ a_2$&$a_3$& $a_4$&$a_5$\\
\hline
    2.3817&
    1.6667&
    1.1664&
    0.8162&
    0.5712\\
\end{tabular}
\end{center}
\end{table}
The no arbitrage condition is satisfied with $R=0.05$.
The historical probability $P$ is in Table \ref{ph}.
\begin{table}[h!]
\begin{center}\caption{}\label{ph}
\begin{tabular}{r |r | r  |r|r }
$p_1$& $p_2$&$p_3$& $p_4$&$p_5$\\
\hline
   0.0829&
    0.1656&
    0.5029&
    0.1658&
    0.0828  \\
\end{tabular}
\end{center}
\end{table}

The  MMEGs and the corresponding risk-neutral European call and put prices for $T=1$ (one-period tree) are in Table \ref{ppG}.
The call price bounds on the whole polytope are $c_{34}=1.9847$ and $c_{15}=7.3789$, the analytical no arbitrage bounds are $1/(1+R)(max(S_0(1+R)-K,0))=0.5858$ and $S_0=21.5831$.
\begin{table}[h!]
\begin{center}\caption{prices corresponding to the MMEGs, maximum price corresponds to $Q_{34}$ and minimum price corresponds to $Q_{15}$.}\label{ppG}
\begin{tabular}{r |r | r  |r }
\text{Measure}&\text{ Probabilities}&\text{Call Prices}&\text{Put Prices}\\
\hline
$Q_{14}$& (0.1493,0,0,0.8507 ,0)&   4.1666 &3.5809\\
$Q_{15}$&(0.2645,0,0,0,0.7355)&      7.3789 &6.7932\\
$Q_{24}$&(0,0.2749,0, 0.7251,0)&    3.6382 &3.0524
\\
$Q_{25}$&( 0,0.4371,0,0,0.5629) &    5.7849 &5.1991
\\
$Q_{34}$&( 0,0,0.677,0.3323,0) &     1.9847  & 1.3990\\
$Q_{35}$&( 0,0,0.8045,0,0.1955) &    2.3915&1.8057
\end{tabular}
\end{center}
\end{table}
Table \ref{tq} provides the MEMM probability $\tilde{Q}$. The binomial trees ($T=2$) corresponding to the minimum and maximum price are shown in Figure \ref{TP2}.

\begin{figure}[h!]
\caption{Trinomial tree. The red line indicates the Binomial tree $3,4$; the green line indicates the Binomial tree $1,5$. The orange point is the overlapping of the two Binomial trees.}\label{TP2}\centering
\includegraphics[width=0.5\linewidth]{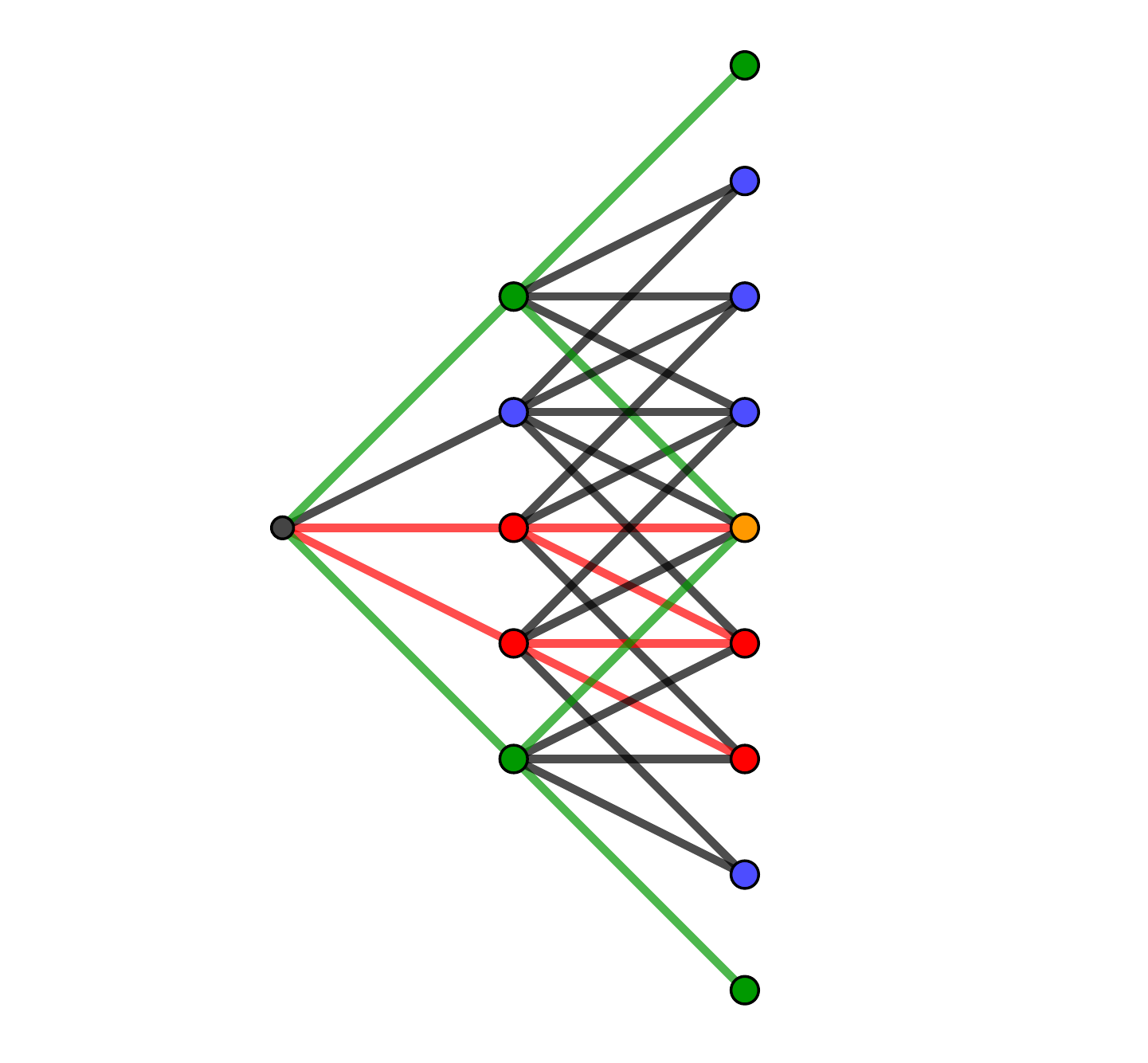}
\end{figure}

%
%
\begin{table}[h!]
\begin{center}\caption{MEMM probability}\label{tq}
\begin{tabular}{r |r | r  |r|r }
$\tilde{q}_1$& $\tilde{q}_2$&$\tilde{q}_3$&$ \tilde{q}_4$&$\tilde{q}_5$\\
\hline
       0.0816&  0.0877&  0.4860& 0.2440& 0.1637
\\
\end{tabular}
\end{center}
\end{table}
Now, for $T=1;2;3;4$ we find the MEMM prices and compare them with the analytical bounds of the entire $\dym$ in
Table \ref{prices}.
\begin{table}[h!]
\begin{center}\caption{European call prices under $\tilde{Q}$ for $T=N=1;2;3,4$ years.}\label{prices}
\begin{tabular}{r |r | r |r||r|r|r}
T&$\tilde{Q}$-call&Min-call&Max-call&$\tilde{Q}$-put&Min-put &Max-put\\
\hline
$1$&3.1245&1.9847&7.3789&2.5387&1.3990&6.7932\\
$2$&4.6914&2.9518&8.3910&3.1079&1.3684&7.3474\\
$3$&5.9427&3.8658&10.6681&3.4090&1.3321&8.1344\\
$4$&7.0503&4.8386&12.4926&3.6117&1.3999&9.0539
\\
\end{tabular}
\end{center}
\end{table}
 It is evident that the range of admissible prices is quite large and therefore in practice it could be more useful consider the range of prices in a class of perturbations of $\tilde{Q}$.
 We define a class $\Gamma(\tilde{Q}, \epsilon)$ of measures in $\dym$ with relative entropy with respect to $\tilde{Q}$ smaller than $\epsilon$. To choose $\epsilon$ we look at the distribution of the relative entropy $I(\tilde{Q}, Q)$ across $\dym$, that is  shown in Figure \ref{relentr5}. The maximal simulated relative entropy is  $I_M(\tilde{Q}, Q)=1.1810$, the minimal is  $I_m(\tilde{Q}, Q)=0$. We choose $\epsilon=0.12$ that corresponds to the $10\%$ of the range spanned by the entropy across $\dym$.

\begin{figure}[h!]
\caption{Relative entropy $I(\tilde{Q},Q)$ across $\dym$.}\label{relentr5}\centering
\includegraphics[width=0.4\linewidth]{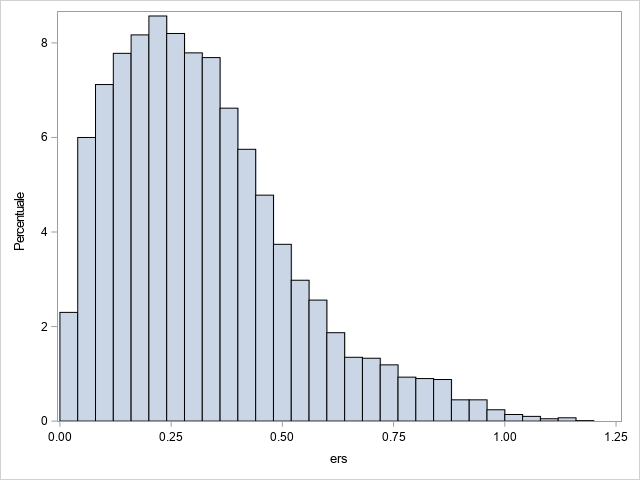}
\end{figure}

 Table \ref{MMboundPer} provides the  maximum and minimum call prices  on the set of  MEMM perturbations $\Gamma(\tilde{Q}, 0.12)$, i.e the class of equivalent martingale measures whose relative entropy with respect to $\tilde{Q}$  is smaller that $0.12$.
\begin{table}[h!]
\begin{center}\caption{Price bounds on $\Gamma(\tilde{Q}, 0.12)$.}\label{MMboundPer}
\begin{tabular}{r|r |r  }
T&Max call price& Min call price \\
\hline
1&    4.2224& 2.5973\\
2&6.0540&3.8092\\
3&7.4961&4.9400\\
4&8.7721&5.9766
\end{tabular}
\end{center}
\end{table}

 Table \ref{MMboundPer} shows that even  small perturbation of the risk-neutral measures have an  impact on prices. However, comparing Table \ref{MMboundPer} and Table \ref{prices} the range spanned by prices across $\Gamma(\tilde{Q}, 0.12)$ is significantly smaller than the range spanned across the whole class.  Properly define the class $\Gamma$ of plausible alternatives could be informative  in real world application.
%
%
%

\section{Conclusion and  future research}\label{credit}

The geometrical representation of discrete distributions  with a given mean presented here improves  our knowledge of discrete models.
In this paper, we show its usefulness for measuring model risk in discrete models and we develop a first application to pricing of derivatives: we consider the  model  risk that arises from a wrong choice of the risk-neutral measure used to price derivatives in incomplete markets, among a class of plausible alternatives.

Part of our ongoing research   is devoted to another application of the geometrical structure of discrete models. We are generalizing  the model risk analysis performed in  \cite{fontana2021model} to the case of unbalanced portfolios.
Consider a credit portfolio $P$ with $d$ obligors.
The components of random variable $\boldsymbol{X}=(X_1,
\ldots, X_{d})$ are the default indicators for the portfolio $P$ and we assume that they have the same  Bernoulli marginal distribution with mean $p$: no assumptions are made on their dependence.

To model the loss of a credit risk portfolio $P$ of $d$ obligors we consider
the sum of the percentage individual losses
\begin{equation*}
L=\sum_{i=1}^{d}w_{i}X_{i},
\end{equation*}
where $w_{i}\in (0,1]$ and $\sum_{1=1}^{d}w_{i}=1$. If we
assume that the weights are given we can measure the  risk associated with  the distribution of joint defaults. In fact, we can move the distribution of $L$ leaving its mean and its support fixed. Let $\mathcal{L}_w$ the class of discrete distributions of the losses for a given $w=(w_1,\ldots, w_d)$.

\begin{definition}
Let $Y$ be a random variable representing a loss with finite mean. Then the $%
\text{VaR}_{\alpha}$ at level $\alpha $ is defined by
\begin{equation*}
\text{VaR}_{\alpha}(Y)=\inf \{y\in \mathbb{R}:P(Y\leq y)\geq \alpha \}.
\end{equation*}%

\end{definition}

The case of equal weights has been discussed in
\cite{fontana2021model}, where analytical  bounds for the VaR are reached on the extremal points.
In this special case the relevant quantity is the number of defaults, a discrete distribution with support on $\{0,\ldots, d\}$. The results discussed in  Section \ref{Sec:distrrisk} allow us to consider unbalanced portfolios and, given the weights $w_i$ that define the support of the loss, we can measure  how by changing the joint  distribution of defaults  in  a given class we affect $\text{VaR}$ of the loss, according to \eqref{bounds1}, that in this case becomes:
\begin{equation}\label{muvar}
\mu_{\mathcal{L}}(VaR, \alpha)=max_{L\in \mathcal{L}}VaR_{\alpha}(L)-min_{L\in \mathcal{L}}VaR_{\alpha}(L).
\end{equation}

%
%

Since  methodology developed in \cite{fontana2021exchangeable} works in high dimension, an empirical investigation on  real data is planned.
After  estimating a distribution for the loss on real data, using for example a Bernoulli mixture model, as  in  \cite{doria2022machine}, we can  consider an unbalanced credit portfolio (e.g. credit cards) and then  we can measure the risk associated with the estimated joint default distribution, as discussed above.

%
%
 \newpage
\appendix
\section{Pentanomial lattice}\label{Pentanomial}
Let $\hat{\mu}$, $\hat{\sigma}$, $\hat{s}$ and $\hat{k}$ be the empirical annualized mean, variance, skewness and excess kurtosis of  returns.
The relevant quantities to construct a pentanomial lattice  are the estimated jump amplitudes $a_i,\, i=1,\ldots, 5$ and historical probability $P$, that are provided in  \cite{yamada2004properties}.  In \cite{yamada2004properties} also $u$ and $d$ are  derived, here we recall the jump amplitudes that are given by:
\begin{equation}\label{Ja}
\left(
\begin{array}
[c]{c}
a_1\\
a_2\\
a_3\\
a_4\\
a_5\\
\end{array}
\right)=\left(
\begin{array}
[c]{c}
\exp\{\hat{\mu}+2\hat{\sigma}\sqrt{1+\frac{\hat{k}}3}\}\\
\exp\{\hat{\mu}+\hat{\sigma}\sqrt{1+\frac{\hat{k}}3}\}\\
\exp\{\hat{\mu}\}\\
\exp\{\hat{\mu}-\hat{\sigma}\sqrt{1+\frac{\hat{k}}3}\}\\
\exp\{\hat{\mu}-2\hat{\sigma}\sqrt{1+\frac{\hat{k}}3}\}\\
\end{array}
\right)
\end{equation}
and the corresponding  historical  probabilities are given by:
\begin{equation}\label{Jp}
\left(
\begin{array}
[c]{c}
p_1\\
p_2\\
p_3\\
p_4\\
p_5\\
\end{array}
\right)=\left(
\begin{array}
[c]{c}
\frac{3+\hat{k}+\hat{s}\sqrt{9+3\hat{k}}}{4(3+\hat{k})^2}\\
\frac{3+\hat{k}-\hat{s}\sqrt{9+3\hat{k}}}{2(3+\hat{k})^2}\\
\frac{3+2\hat{k}}{2(3+\hat{k})}\\
\frac{3+\hat{k}+\hat{s}\sqrt{9+3\hat{k}}}{4(3+\hat{k})^2}\\
\frac{3+\hat{k}-\hat{s}\sqrt{9+3\hat{k}}}{4(3+\hat{k})^2}\\
\end{array}
\right),
\end{equation}
For the construction of the lattice,  the up and down rates and the derivation of the above quantities see \cite{yamada2004properties}.

\section{Uniform sampling}\label{US}
To simulate a functional on a class $\mathcal{P}\subseteq \dym$ we  use uniform sampling following \cite{fontana2021exchangeable}.
We  partition $\dym$ into simplices  $\mathcal{T}_i, \, i\in \mathcal{I}$ (e.g. using a Delaunay triangulation)
where $\mathcal{I}$ is a proper set of indices, and $\mathcal{T}_i \cap \mathcal{T}_j = \emptyset$ for $i \neq j$.
Let $\phi(y)=\phi_d(\xx)$, where $y=\sum_{i=1}^d x_i$, $y=\{0,\ldots, d\}$. Let us denote by $F_\phi$ the distribution of $E^P[\phi(Y)]$, where $P\in\mathcal{P}$. We get
\begin{equation}\label{prob_tot}
F_\phi(t)=P(E^P[\phi(Y)] \leq t)=\sum_{i \in \mathcal{I}}P( \mathcal{T}_i) P(E^P[\phi(Y)] \leq t| \mathcal{T}_i).
\end{equation}

If we assign a uniform measure on the space $\mathcal{P}$ the probability $P( \mathcal{T}_i) $ of sampling a probability mass function in the simplex $\mathcal{T}_i$ is simply the ratio between the volume of  $\mathcal{T}_i$ and the total volume of $\mathcal{P}$, i.e.
\begin{equation} \label{volsim}
P(\mathcal{T}_i)=\frac{\text{vol}(\mathcal{T}_i)}{ \text{vol}(\mathcal{\mathcal{P}})}.
\end{equation}
The volume of each $\mathcal{T}_i$ and consequently that of $\mathcal{P}$ can be easily computed because the volume  of an $n$-simplex in $n$-dimensional space with vertices $(v_0, \ldots, v_n)$ is
\[
{\displaystyle \left|{1 \over n!}\det {\begin{pmatrix}v_{1}-v_{0},&v_{2}-v_{0},&\dots ,&v_{n}-v_{0}\end{pmatrix}}\right|}
\]
where each column of the $n \times n$ determinant is the difference between the vectors representing two vertices \cite{stein1966note}.

The probability $P(E[\phi(Y)] \leq t| \mathcal{T}_i)$ is the ratio between the volume of the region
 \begin{equation}\label{region}\mathcal{R}_{i,t}=\{p_Y \in \mathcal{T}_i: E[\phi(Y)] \leq t\}\end{equation}
 and the volume of $\mathcal{T}_i$, i.e.
\begin{equation} \label{volfrustum}
P(E[\phi(Y)] \leq t| \mathcal{T}_i)=\frac{\text{vol}(\mathcal{R}_{i,t})}{\text{vol}(\mathcal{T}_i)}.
\end{equation}
The computation of the volume of $\mathcal{R}_{i,t}$ will depend on the definition of $\phi$ in the expectation measure $E[\phi(Y)]$.
We
 uniformly sample at random over  $\mathcal{T}_i$ and determining the relative frequency of the points that fall in the region $\mathcal{R}_{i,t}$, as defined in \eqref{region}
\[
\widehat{\left( \frac{\text{vol}(\mathcal{R}_{i,t})}{\text{vol}(\mathcal{T}_i)} \right)}=\frac{\#\{p_k \in \mathcal{R}_{i,t}, \, k=1,\ldots, N \}}{N},
\]
where $N$ is the size of the sample. In these cases an estimate $\hat{F_\phi}$ of the distribution $F_\phi$ will be obtained.

\bibliographystyle{ieeetr}
\bibliography{biblio}

\begin{thebibliography}{10}

\bibitem{breuer2016measuring}
T.~Breuer and I.~Csisz{\'a}r, ``Measuring distribution model risk,'' {\em
  Mathematical Finance}, vol.~26, no.~2, pp.~395--411, 2016.

\bibitem{bellini2021law}
F.~Bellini, P.~Koch-Medina, C.~Munari, and G.~Svindland, ``Law-invariant
  functionals on general spaces of random variables,'' {\em SIAM Journal on
  Financial Mathematics}, vol.~12, no.~1, pp.~318--341, 2021.

\bibitem{fontana2021model}
R.~Fontana, E.~Luciano, and P.~Semeraro, ``Model risk in credit risk,'' {\em
  Mathematical Finance}, vol.~31, no.~1, pp.~176--202, 2021.

\bibitem{cox1979option}
J.~C. Cox, S.~A. Ross, and M.~Rubinstein, ``Option pricing: A simplified
  approach,'' {\em Journal of financial Economics}, vol.~7, no.~3,
  pp.~229--263, 1979.

\bibitem{kamrad1991multinomial}
B.~Kamrad and P.~Ritchken, ``Multinomial approximating models for options with
  k state variables,'' {\em Management science}, vol.~37, no.~12,
  pp.~1640--1652, 1991.

\bibitem{ssebugenyi2013minimal}
C.~S. Ssebugenyi, I.~J. Mwaniki, and V.~S. Konlack, ``On the minimal entropy
  martingale measure and multinomial lattices with cumulants,'' {\em Applied
  Mathematical Finance}, vol.~20, no.~4, pp.~359--379, 2013.

\bibitem{cont2006model}
R.~Cont, ``Model uncertainty and its impact on the pricing of derivative
  instruments,'' {\em Mathematical {F}inance}, vol.~16, no.~3, pp.~519--547,
  2006.

\bibitem{de1997computational}
M.~De~Berg, M.~Van~Kreveld, M.~Overmars, and O.~Schwarzkopf, {\em Computational
  geometry}.
\newblock Springer, 1997.

\bibitem{fontana2021exchangeable}
R.~Fontana and P.~Semeraro, ``Exchangeable bernoulli distributions: High
  dimensional simulation, estimation, and testing,'' {\em Journal of
  Statistical Planning and Inference}, vol.~225, pp.~52--70, 2023.

\bibitem{shaked2007stochastic}
M.~Shaked and J.~G. Shanthikumar, {\em Stochastic orders}.
\newblock Springer, 2007.

\bibitem{denuit1998optimal}
M.~Denuit and C.~Vermandele, ``Optimal reinsurance and stop-loss order,'' {\em
  Insurance: Mathematics and Economics}, vol.~22, no.~3, pp.~229--233, 1998.

\bibitem{eberlein1997range}
E.~Eberlein and J.~Jacod, ``On the range of options prices,'' {\em Finance and
  Stochastics}, vol.~1, no.~2, pp.~131--140, 1997.

\bibitem{frittelli1995minimal}
M.~Frittelli, ``Minimal entropy criterion for pricing in one period incomplete
  markets,'' {\em Dept. Metodi Quantitativi Tech. Report}, no.~99, 1995.

\bibitem{ekstrom2004properties}
E.~Ekstr{\"o}m, ``Properties of {A}merican option prices,'' {\em Stochastic
  Processes and their Applications}, vol.~114, no.~2, pp.~265--278, 2004.

\bibitem{yamada2004properties}
Y.~Yamada and J.~A. Primbs, ``Properties of multinomial lattices with cumulants
  for option pricing and hedging,'' {\em Asia-Pacific Financial Markets},
  vol.~11, no.~3, pp.~335--365, 2004.

\bibitem{doria2022machine}
M.~Doria, E.~Luciano, and P.~Semeraro, ``Machine learning techniques in joint
  default assessment,'' {\em arXiv preprint arXiv:2205.01524}, 2022.

\bibitem{stein1966note}
P.~Stein, ``A note on the volume of a simplex,'' {\em The American Mathematical
  Monthly}, vol.~73, no.~3, pp.~299--301, 1966.

\end{thebibliography}

\end{document}